\documentclass[journal,draftclsnofoot,onecolumn,12pt]{IEEEtran}
\usepackage{relsize}
\usepackage{amsmath,amssymb,amsthm}
\usepackage{mathrsfs}
\usepackage{epsfig}
\usepackage{enumerate}
\usepackage[lmargin=1.1in,rmargin=1.1in,tmargin=1.1in,bmargin=1.1in]{geometry}
\usepackage[latin1]{inputenc}
\usepackage{tikz}
\usepackage[latin1]{inputenc}
\usetikzlibrary{shapes,arrows}
\usepackage{float}
\usepackage{verbatim}
\usepackage{textcomp}
\usetikzlibrary{calc}
\usetikzlibrary{arrows}
\usetikzlibrary{plotmarks}
\usetikzlibrary{fadings}
\usetikzlibrary{shadings}
\usepackage{setspace}
\usepackage{url}
\title{Capacity and Spectral Efficiency of Interference Avoiding Cognitive
Radio
with Imperfect Detection}

\author{Aaqib Patel*,~\IEEEmembership{Student Member,~IEEE,} Md. Zafar Ali
Khan,~\IEEEmembership{Member,~IEEE,}
        S. N. Merchant*, U. B. Desai,~\IEEEmembership{Senior Member,~IEEE,}%
\thanks{This work was supported through grants to
Mohammed Zafar Ali Khan and S. N. Merchant by the
DIT grant on High performance cognitive radio.}% <-this % stops a space
\thanks{Md. Zafar Ali Khan and  U. B. Desai are with the  Dept. of Electrical Engineering, Indian Institute of Technology, Hyderabad
AP 502205, India. Email: zafar@iith.ac.in, ubdesai@iith.ac.in. Aaqib Patel* and S. N. Merchant* are with the  Dept. of Electrical Engineering, Indian Institute of Technology, Bombay; Email: aaqib,merchant@ee.iitb.ac.in,
}}
\begin{document}
\maketitle
\doublespace
\newtheorem{myprop}{Propositon}
\newtheorem{mydef}{Definition}
\newtheorem{mythe}{Theorem}
\newtheorem{mylem}{Lemma}
\newtheorem{myex}{Example}
\newtheorem{mycor}{Corollary}
\renewcommand*\thesection{\Roman{section}}
\begin{abstract}
In this paper, we consider a model in which the unlicensed or the Secondary User
(SU) equipped with a Cognitive Radio (CR) (together referred to as CR)
interweaves its transmission with that of the licensed or the Primary User
(PU). In this model, when the CR detects the PU to be (i) busy it does not
transmit and; (ii) PU to be idle it transmits. Two situations based on CR's
detection of PU are considered, where the CR detects PU (i) perfectly - referred
to as the ``ideal case'' and; (ii) imperfectly - referred to as ``non ideal
case''. For both the cases we bring out the rate region, sum capacity of PU and
CR and spectral efficiency factor - the ratio of sum capacity of PU and CR to
the capacity of PU without CR. We consider the Rayleigh fading channel to
provide insight to our results. For the ideal case we study the effect of PU
occupancy on spectral efficiency factor. For the non ideal case, in addition to
the effect of occupancy, we study the effect of false alarm and missed detection
on the rate region and spectral efficiency factor. We characterize the set
of values of false alarm and missed detection probabilities for which the system
benefits, in the form of admissible regions. We show that false alarm has a more
profound effect on the spectral efficiency factor than missed detection. We also
show that when PU occupancy is small, the effects of both false alarm and missed
detection decrease. Finally, for the standard detection techniques viz. energy
detection, matched filter and magnitude squared coherence, we show that that the
matched filter performs best followed by magnitude squared coherence followed by
energy detection with respect to spectral efficiency factor.
\end{abstract}
\section{Introduction}
While certain parts of the frequency spectrum are crowded by users, most part of
the spectrum still remains heavily unoccupied \cite{fcc}. Due to this
\textit{spectrum imbalance}, the legacy command and control policy
of the FCC and the recent proliferation in usage of wireless devices, new
techniques that can perform Dynamic Spectrum Access (DSA) are needed. The so
called Cognitive Radio (CR), in essence, performs DSA by sensing the radio
environment, choosing the best channel (frequency, time slot or codes)
that is available and adapting parameters to communicate simultaneously
ensuring the interference to legacy users below a prescribed level \cite{mito}.

In literature (for example \cite{book}, \cite{goldd}) there are three paradigms
for CR based systems mentioned, based upon the nature of operation of the CR,
viz. (i) interference avoiding or interweave, (ii) interference limited or
underlay and (iii) interference mitigating or overlay. Each of the above
paradigm assumes different amounts or levels of side information that the SU has
about PU transmission.

From an information theoretic viewpoint, most of the recent works focus on
modelling the CR channel as an interference channel with some side information
about the PU codebooks, messages, channel characteristics etc, being available
to the CR transmitter \cite{goldd}-[15]. Also they mostly address
underlay and overlay based CRNs.

In \cite{jaf}  a two switch model for CR is presented with one switch at the SU
transmitter and the other at the SU receiver. The role of each switch is to
identify transmitting opportunity (channel) which can be different at the
SUs transmitter and receiver. Transmission is then done on a channel which
is identified to be free on both the transmitter and the receiver ends. In
\cite{nat} the CR channel is modeled as two sender two receiver model with the
transmitters having non causal information. In \cite{vish1},\cite{vish2}
capacity for low and high interference regimes is analyzed, wherein the PU
codeword is assumed to be known to the CR. Extending from two user case to the three user case, in \cite{nag} achievable rate regions and outer bounds are derived. Channels are modeled as interference channels where the transmitters cooperate in a unidirectional manner via 3 different noncausal message-sharing mechanism specified therein. 

For the case of underlay based systems, in \cite{mus},\cite{na} the capacity gains offered when only partial channel information of the link between the CR transmitter and PU receiver is available to the CR under average received power constraint is studied. In \cite{son} an optimal power allocation policy to maximize capacity, which exploits a two-dimensional frequency-selectivity on CRs own channel as well as CR to PU channel under an interference-power constraint as well as a conventional transmit power constraint is proposed. Moving further ahead \cite{ban} proposes a joint overlay and underlay technique to maximize the total capacity of CR.

While most of the above works either attempt to characterize the achievable regions or maximize CR capacity, certain works study about the sum capacity of the CR and PU together so that capacity gains of entire system is brought out \cite{lia},\cite{zha}.

Significantly, all the above works assume that the CR and PU have perfect side
information about the PU occupancy. Although such models are of use, practical scenarios wherein there
is imperfect side information available needs to be studied as capacity can
decrease with imperfect side information, significantly.

In \cite{book}, a simple rate region involving rates of PU and CR has been
brought out for the case where the CR interweaves its communication with the
PU by means of time sharing. It is assumed therein that CR has perfect
knowledge of PU's channel occupancy. We refer to this as the ``ideal case''.

In this paper, we find the \textit{spectral efficiency factor}, which is the
ratio of the sum capacity of CR and PU to the capacity of the PU in absence of
CR, for the ideal case. With the help of an example of Rayleigh flat fading
channel, we study its variation with PU occupancy for different relative power
levels of the CR and PU. Along with the ideal case, we consider the case where
CR interweaves its communication with PU with ``only an imperfect'' knowledge of
the PUs channel occupancy. We refer to this as the ``non ideal case''. We obtain
rate regions for the non ideal case and compare it with the ideal case. We show
that the rate region of the non ideal case is a sub-region of the ideal case. We
also discuss the effects of false alarm and missed detection on the rate region
for the non ideal case. We then find the maximum sum capacity and use this
expression compute the spectral efficiency factor. We bring out
effects of false alarm and missed detection on the spectral efficiency factor.
We show that, probability of false alarm has more profound effect than the
missed detection on spectral efficiency factor.

In the ideal case the, spectral efficiency factor is always greater than unity.
However, in the non ideal case, if the error in detection is large, i.e. the
probability of false alarm, $p_{fa}$, and the probability of missed detection,
$p_{md}$, either individually or collectively transgress certain limits; the
spectral efficiency factor turns out to be less than unity. We characterize
regions made up of those $(p_{fa}, p_{md})$ points, called the admissible
regions, for which we obtain spectral efficiency factor greater than unity.
Finally we compare known detection schemes from a spectral efficiency
perspective.

The rest of the paper is as follows. In Section \ref{sec2} we discuss the system
model. In Section \ref{sec3} we show the rate region for the case where the
CR detects the PU perfectly. We then address the non ideal detection and
characterize the corresponding  rate region in Section \ref{sec4}. We also find
conditions for obtaining maximum sum capacity therein. Then we discuss
the spectral efficiency factor in Section \ref{sec5}, wherein we discuss the
effects of false alarm and missed detection on the spectral efficiency factor.
We then show a comparison of the certain standard detection techniques in light
of the results obtained in section \ref{sec6}. Finally conclusions are given in
section \ref{sec7}.
\section{System Model}\label{sec2}
In this paper we assume that the Primary User (PU) transmits with a probability
$1-p$ and the channel is free with a probability $p$. However, from a detection
perspective, the occupancy of a single channel is defined as the unconditional
probability that the measured signal strength does exceeds some predetermined
threshold \cite{Gib}. In spectrum sensing, the detection involves comparing a
test statistic against a threshold, suitably chosen. Let $Y$ be the observed
random variable and $X$ the transmit random variable, so that $X=1$ if there is
no transmission and 0 otherwise. It follows that $p = \mathbb{P}[X= 1]$. We now
discuss the system models for the two cases  based on the
relation of $Y$ and $X$.
\subsection{System Model for Ideal Detection}\label{ssec21}
In case of perfect detection, $Y = X$. This is referred to as the case
of ``ideal'' detection. The system model is described in Fig. \ref{idcha}.

The baseband relations for the PU and the CR are defined as
\begin{align} \label{ideq}
\mbox{For, }Y=X=0; \quad Y_p = h_pX_p + Z_p, \nonumber \\
\mbox{For, }Y=X=1; \quad Y_c = h_cX_c + Z_c,
\end{align}
where $X_p$ and $X_c$ are the signals transmitted, $Y_p$ and
$Y_c$ are the signals received by the PU and the CR respectively. $Z_p$ and
$Z_c$ are uncorrelated circularly symmetric zero mean variance $\sigma^2$ white
Gaussian random variables i.i.d. over time for the PU and CR respectively,
denoted as, $Z_p, Z_c \sim \mathcal{CN}(0,\sigma^2)$. The average power
constraints on the PU and CR are $\mathbb{E}[X_p]\leq P_p$ and
$\mathbb{E}[X_c]\leq P_c$. The fading coefficients $h_p$ and $h_c$ are
uncorrelated and i.i.d over time. These channels are asymmetrically time shared,
i.e. the CR channel comes into existence only when the PU channel is not used.
\subsection{System Model for Non Ideal Detection} \label{ssec22}
In general, $Y\neq X$, which is referred to as the case of ``non ideal''
detection. In non ideal detection there are errors in detection. The probability
of false alarm, $p_{fa} = \mathbb{P}[Y=1|X=0]$ is the probability that CR
detects the channel to be free given it was occupied by the PU. Similarly, the
probability of missed detection, $p_{md} = \mathbb{P}[Y=0|X=1]$, is the
probability that CR detects the channel to be occupied given it was free from
PU transmission \cite{kay}. Note that the other two conditional probabilities
can be expressed in terms of $p_{fa}$ and $p_{md}$. The system model for non
ideal case is as shown in Fig. \ref{nidcha}.

Due to the possibility of false alarm, the PU and CR might transmit at the same
time, which leads to a possibility of the interference channel. Hence, the
channel can be a pure flat fading channel with probability $\mathbb{P}[X=1,Y=1]$
and an interference channel with probability $\mathbb{P}[X=0,Y=1]$. The entire
channel is time shared between a pure flat fading channel for PU, a pure flat
fading channel for the CR and finally an interference channel for the PU and CR.
We thus have a time sharing random variable as follows.
\begin{mydef}\label{time}
Time sharing random variable $Q$ is as defined
\begin{enumerate}
 \item $Q=1$ is the event that $X=1,Y=1$, CR correctly detects the channel to
be free;
 \item $Q=2$ is the event that $X=0,Y=1$, CR erroneously detects the channel to
be free;
 \item $Q=3$ is the event that $X=1,Y=0$, CR erroneously detects the channel to
be occupied;
 \item $Q=4$ is the event that $X=0,Y=0$, CR correctly detects the channel to
be occupied.
\end{enumerate}
\end{mydef}
The probability distribution on $Q$ is as follows
\begin{align}
\mathbb{P}(Q = 1)\ & = \mathbb{P}[X=1,Y=1] = p(1-p_{md}), \nonumber \\
\mathbb{P}(Q = 2)\ & = \mathbb{P}[X=0,Y=1] = (1-p)p_{fa}, \nonumber \\
\mathbb{P}(Q = 3)\ & = \mathbb{P}[X=1,Y=0] = pp_{md}, \nonumber \\
\mathbb{P}(Q = 4)\ & = \mathbb{P}[X=0,Y=0] = (1-p)(1-p_{fa}).
\end{align}
Note that $Q=1,2,4$ respectively denote the instances where there is a flat
fading channel for the CR, an interference channel between the PU and the CR,
and a flat fading channel for the PU. Also note that $Q=3$ is the instance where
the channel goes waste as there is no transmission during this time, neither
from PU nor from CR. The distribution of the time sharing random variable is
fixed for a given $(p_{fa},p_{md})$. We have the baseband representation as
\begin{align} \label{nideq}
 \mbox{For }Q=1;\ & \quad Y_c = h_cX_c + Z_c, \nonumber \\
 \mbox{For }Q=2; \ & \quad Y_c = h_cX_c + h_{pc}X_p + Z_c \mbox{ and } Y_p =
h_pX_p + h_{cp}X_c + Z_p,  \nonumber \\
 \mbox{For }Q=4; \ & \quad Y_p = h_pX_p + Z_p,
\end{align}
where $Y_p$, $Y_c$, $X_p$, $X_c$, $Z_p$, $Z_c$, $h_p$ and $h_c$ remain the same
as in the ideal case. $h_{pc}$ and $h_{cp}$ are also i.i.d over time and
uncorrelated with each other as well as with $h_p$ and $h_c$.

The system model for the non ideal case is a very general model and encompasses some of the models available. In Fig. \ref{gen1}, we have shown how the two switch model proposed in \cite{jaf} can be expressed as a special case of the model we have proposed. There are three possibilities, viz (i) the CR transmitter and the receiver detect the same channels to be free. Here the CR transmitts using perfect knowledge and this is the case of ideal detection i.e. $p_{fa}=p_{md}=0$, (ii) the CR tranmsitter finds certain channels to be free which are not found free at the CR receiver hence if the CR transmits then it interfers with the PU which implies $p_{fa}\neq0$, and (iii) the CR receiver determines free channels but the CR transmitter does not which leads to a loss in transmission opportunity an hence, $p_{md}\neq0$. 
\section{Rate Region and Spectral Efficiency Factor for Ideal
Detection}\label{sec3}
In this section we present the rate region \cite{book} and spectral efficiency
factor for ideal detection. These results help in presenting the non ideal case
in proper perspective. We then define spectral efficiency factor as the ratio of
the sum capacity of the CR and PU to the capacity of PU in absence of CR. We
discuss its variation as a function of $p$ for different values of relative
transmit power levels of CR and PU with help of an example.

The system model for the ideal case is described in subsection \ref{ssec21}. The
average capacity for the PU, $C_p$, and CR, $C_c$, is given by \cite{tse}
\begin{align} \label{sim}
 C_p = \mathbb{E}_{|h_p|}
\left[\log\left(1+\dfrac{|h_p|^2P_p}{\sigma^2}\right)\right]
\mbox{bits/complex dimension} \nonumber \\
 C_c = \mathbb{E}_{|h_c|}
\left[\log\left(1+\dfrac{|h_c|^2P_c}{\sigma^2}\right)\right]
\mbox{bits/complex dimension}.
\end{align}
The corresponding rate region is given by
\begin{align} \label{idreq}
\ & \mathcal{C} = \{(R_c,R_p)| 0 \leq R_c \leq pC_c, 0\leq R_p
\leq (1-p)C_p,  \mbox{ }0\leq p \leq 1 \}
\end{align}
where the rates $R_c$ of the PU transceiver pair and $R_p$ of the CR
transceiver pair is achieved through ideal white space filling
\cite{book}. Fig. \ref{idrate} shows the rate region.

The sum capacity of PU and CR is $pC_c + (1-p)C_p$. There are three possible
cases viz. (i) $C_p>C_c$, in this case $p=0$ maximizes the sum capacity,
(ii) $C_p<C_c$, in this case $p=1$ maximizes the sum capacity, and (iii)
$C_p=C_c$, in this case any value of $p$ would maximize the sum capacity.

We now define spectral efficiency factor in terms of the sum capacity as
\begin{mydef} \label{speeff}
Spectral efficiency factor, $\eta$, is defined as the ratio of sum capacity of
CR and PU to the capacity of PU in the absence of CR. Formally,
$$
\eta = \dfrac{D_p^{CR} + D_s^{CR}}{D_p}
$$
where,
\begin{itemize}
\item $D_p^{CR}$ = PUs data rate achievable with CR present,
\item $D_s^{CR}$ = CRs achievable data rate
\item $D_p$ = PUs data rate achievable without CR present.
\end{itemize}
\end{mydef}
With this definition, we have
\begin{mythe}
The Spectral efficiency factor for ideal detection with average power
constraints as $P_c$ and $P_p$ for the CR and PU respectively is
given by
$$\eta = 1+\dfrac{pC_c}{(1-p)C_p}$$
where $C_c$ and $C_p$ are as given in  (\ref{sim}).
\end{mythe}
\begin{proof}
This follows from definition \ref{speeff}. and (\ref{sim}).
\end{proof}
The value of this factor is always greater than unity as the CR transmits when
the primary is idle, making more use of resource. It is of interest to see the
variation of $\eta$ with $p$.
\begin{mycor}\label{cor1}
 $\eta$ increases with $p$ with a rate of $\dfrac{1}{(1-p)^2}.\dfrac{C_c}{C_p}$.
\end{mycor}
\begin{proof}
As $p$ increases, then $(1-p)$ decreases which implies that
$\dfrac{1}{1-p}$ increases. Hence $\dfrac{p}{1-p}$ increases with $p$. Since
$C_p$ and $C_c$ are independent of $p$, $\eta$ increases with $p$. By
differentiating $\eta$ w.r.t. $p$, we get the slope of the curve of $\eta$.
\end{proof}

We look at an example where fading is Rayleigh to better understand the
consequence of Corollary \ref{cor1}.
\begin{myex}
Consider, the fading coefficients $h_p$ and $h_c$ to be uncorrelated and i.i.d
over time, circularly symmetric zero mean complex Gaussian random variables with unit variance, denoted as
$h_p, h_c \sim \mathcal{CN}(0,1)$. Thus $|h_p|, |h_c|$ are Rayleigh distributed
with parameter value 1, denoted as $|h_p|, |h_c| \sim \mbox{Rayleigh}(1)$. We assume
a fixed $P_p = 10000$ mW and we take 4 different values of $P_c$ such that $RS =
10\log\dfrac{P_p}{P_c} = \{0,10,20,30\}$ dB. We make plot of the $\eta$ vs $p$
for these 4 different cases. Note that throughout the paper $RS$ has the same
definition unless otherwise specified.

We observe from Fig. \ref{exm}, that for higher values of $p$, $\eta$
increases rapidly irrespective of $RS$. This is clear from Corollary \ref{cor1}
since $\dfrac{1}{(1-p)^2}$ increases rapidly, faster than increase in $p$ and
$RS$ does not effect $p$. Hence, even if $p$ increases slightly $\eta$ increases
rapidly. Looking at Fig. \ref{exm}, for $RS = 10$ dB, when we move from
$p=0.8$ to $p=0.98$, the spectral efficiency factor increases from 5 to
around 50 which is a ten times increase. Comparing this to $RS = 0$ dB, the
growth is more rapid from $p=0.8$ to $p=0.98$ there is 20 times increase. Hence,
for higher $RS$ the growth is more rapid. To study the effect of $\eta$ for
small value of $p$, i.e. $p\in[0, 0.5]$ we blow up a portion of Fig. 4 in
Fig. 5. Observe when $RS$ is less, irrespective of $p$, the growth rate of
$\eta$ is much higher than the case when $RS$ is higher. This is evident from
Corollary \ref{cor1} since lower $RS$ would imply higher $\dfrac{C_c}{C_p}$.
\end{myex}
Note that while $\eta$ for the ideal case is greater than unity, in the non
ideal case it may not always be greater than unity as we see in the
Section~\ref{sec5}.

Note that all the Figs. \ref{pfrate} - \ref{ratepmpf} and
\ref{pflos}-\ref{snrdet26} are plotted for Rayleigh fading channels with unit
second moment. That is all $h_c$, $h_p$, $h_{cp}$ and $h_{pc}$  are complex
Gaussian random variables with zero mean and unit variance. However, since we
consider average or ergodic capacity, the analysis and conclusions are valid for
any other distribution for which expectation is finite. Also note that the average received SNR at the PU receiver is assumed to be $0dB$ unless otherwise specifically mentioned.
\section{Rate Region and maximum sum capacity for Non Ideal Detection}
\label{sec4}
In this section, we find out the rate region for the non ideal case for which
the system model is discussed in section \ref{ssec22}. By
considering all the cases defined in the time sharing random variable $Q$ in
definition \ref{time}, we bring out the capacity expressions and then specify
the rate region. We look at the effects of $p_{fa}$, $p_{md}$ and
$(p_{fa},p_{md})$ together on the rate region with the help of plots. We see
that the rate region approaches that of the ideal case as $(p_{fa},p_{md})$ tend
to zero. We then write the sum capacity and maximize it under two specific type
of constraints viz, when $p$ is specified and the maximization is over
$(p_{fa},p_{md})$. The other  case is when $(p_{fa},p_{md})$ are specified
and maximization is over $p$.

We have,
\begin{mylem}\label{rate1}
The average capacity $C_p'$ and $C_c'$ of the PU and CR for average transmit
power constraints $P_p$ and $P_c$ are respectively as follows
\begin{align}
C_p' \ & = (1-p)((1-p_{fa})A_p+p_{fa}B_p)), \\
C_c' \ & = p(1-p_{md})A_c + (1-p)p_{fa}B_c,
\end{align}
where,
\begin{align}
 A_p \ & = \mathbb{E}_{|h_p|}\left[\log\left(1+\dfrac{|h_p|^2P_p}{
\sigma^2 } \right)\right], \nonumber \\
 B_p \ & = \mathbb{E}_{|h_{pp}|}\left[\log\left(1+\dfrac{|h_{pp}|^2P_p}{
P_c + \sigma^2 } \right)\right], \nonumber \\
 A_c \ & = \mathbb{E}_{|h_c|}\left[\log\left(1+\dfrac{|h_c|^2P_c}{
\sigma^2 } \right)\right], \nonumber \\
 B_c \ & = \mathbb{E}_{|h_{cc}|}\left[\log\left(1+\dfrac{|h_{cc}|^2P_p}{
P_p + \sigma^2 } \right)\right]. \nonumber
\end{align}
The units are again bits/complex dimensions.
\end{mylem}
\begin{proof}
We consider the 4 different situations described by $Q$.
\begin{itemize}
 \item For $Q=1$ we have the baseband equations as given in (\ref{nideq}) and
reproduced here, $$Y_c = h_cX_c + Z_c.$$ Hence, the capacity
of CR given $Q=1$ is $$C_c|_{(Q=1)} =
\mathbb{E}_{|h_{c}|}\left[\log\left(1+\dfrac{|h_{c}|^2P_c}{\sigma^2 }
\right)\right].$$ Note that $C_p|_{(Q=1)}=0$.
 \item For $Q=2$ we have the baseband equations as given in (\ref{nideq}) and
reproduced here,
$$Y_c =  h_{c}X_c + h_{pc}X_p + Z_c, \quad Y_p = h_{p}X_p +
h_{cp}X_c + Z_p.$$ Here, both the PU and CR treat each others transmissions as
Gaussian noise with mean 0 and variance $\mathbb{E}[|h_{pc}|^2P_p] = P_p$ and
$\mathbb{E}[|h_{cp}|^2P_c] = P_c$ respectively. Hence, the total noise variance
for the PU and CR is $P_p +\sigma^2$ and $P_c + \sigma^2$ respectively. Hence,
capacity of the PU and CR is $$
C_p|_{(Q=2)} = \mathbb{E}_{|h_{p}|}\left[\log\left(1+\dfrac{|h_{p}|^2P_p}{
P_c + \sigma^2 } \right)\right], \mbox{ }
C_c|_{(Q=2)} = \mathbb{E}_{|h_{c}|}\left[\log\left(1+\dfrac{|h_{c}|^2P_c}{
P_p + \sigma^2 } \right)\right].
$$
 \item For $Q=3$. We have $C_c|_{(Q=3)} = C_p|_{(Q=3)} = 0$.
 \item For $Q=4$ we have the baseband equations as given in equation
(\ref{nideq}) $$Y_p = h_pX_p + Z_p.$$ Hence, the capacity of PU given $Q=4$
is $$C_p|_{(Q=4)} =
\mathbb{E}_{|h_{p}|}\left[\log\left(1+\dfrac{|h_{p}|^2P_p}{\sigma^2 }
\right)\right].$$
Note that $C_c|_{(Q=4)}=0$.
\end{itemize}
Finally we have, the capacity for CR and PU as
\begin{align}
 C_c' \ & = \sum_{i=1}^4C_c|_{(Q=i)}P(Q=i) \nonumber \\
 \ & = (1-p)((1-p_{fa})A_p+p_{fa}B_p)) \\
 C_p' \ & = \sum_{i=1}^4C_p|_{(Q=i)}P(Q=i) \nonumber \\
 \ & = p(1-p_{md})A_c + (1-p)p_{fa}B_c
\end{align}
These follow from the distribution of the time sharing random variable defined
in the system model in subsection \ref{ssec22}.
\end{proof}
Note that the values of $A_p$, $B_p$, $A_c$ and $B_c$ as specified in the Lemma
\ref{rate1}, will be used throughout the paper unless otherwise specified.

The corresponding rate region is given by
\begin{align} \label{nidreq}
\hat{\mathcal{C}}(p_{fa},p_{md}) = \{ \ & (R_c,R_p)| 0\leq R_p \leq
(1-p)C_2(p_{fa}),\nonumber \\ \ &
0\leq R_c \leq (1-p)C_1'(p_{fa}) + pC_1(p_{md}),
 \mbox{ }0\leq p \leq 1 \}
\end{align}
where the rates, $R_p$ of the PU transceiver pair and $R_c$ of the CR
transceiver pair are achieved through non ideal white space filling. Also,
$C_1(p_{md}) = (1-p_{md})A_c$, $C_1'(p_{fa}) = p_{fa}B_c$ and
$C_2(p_{fa}) = (1-p_{fa})A_p + p_{fa}B_p$. Note that,
$$\hat{\mathcal{C}}(p_{fa},p_{md})\arrowvert_{(p_{fa}=p_{md}=0)} =
\mathcal{C}.$$ where $\mathcal{C}$ is the rate region of the ideal case,
specified in \ref{idreq}. In other words, $\hat{\mathcal{C}}\rightarrow
\mathcal{C}$ as $p_{fa},p_{md}$ tends to zero. This implies that for smaller
values of $(p_{fa},p_{md})$, the rate region in the non ideal case comes closer
to that of the ideal case.
Consider the rate region shown in Fig. \ref{nidrate}, the interior of the red diagonal
line along with the axes, specifies the rate region of the ideal detection and
interior of the blue line along with the axes, specifies the rate region for the
non ideal detection. Clearly, the non ideal case is the sub-region of the ideal
case as it completely lies inside the ideal case region. Moreover, the cut on
the vertical (or PU) axis, given by $C_2(p_{fa})$, is purely due to $p_{fa}$,
while that on the
horizontal (or CR) axis is due to $p_{md}$ and $p_{fa}$ both (given by
$C_1(p_{md})$, $C_1'(p_{fa})$ respectively) as
justified from the expression of $C_p'$ and $C_c'$ specified in Lemma
\ref{rate1}. Also note that for a $R_c<C_1'(p_{fa})$ there is no decrease in $R_p$.

Fig. \ref{pfrate} shows that the rate region for different values of $p_{fa}$ when the
$p_{md}$ is constant. Note that the non-ideal rate region approaches the ideal
rate region as $p_{fa}$ decreases, along the $Y$ axis only. This is because the
$p_{fa}$ effects the PU more than the CR. Fig. \ref{pmrate} shows that
the rate region for different values of $p_{md}$ when the
$p_{fa}$ is constant. Observe the change in rate region is only on the $X$ axis.
Fig.  \ref{ratepmpf} gives the rate region when both $p_{fa}$ and $p_{md}$ vary simultaneously.
\subsection{Maximum sum capacity}
The sum capacity of CR and PU is given by
\begin{align}
 C_c' + C_p' = \ & p(1-p_{md})A_c +
 (1-p)p_{fa}B_c %\nonumber \\ \ &
 + (1-p)((1-p_{fa})A_p+p_{fa}B_p) \nonumber \\
 = \ & -A_cpp_{md} - (1-p)(A_p-B_p-B_c)p_{fa} %\nonumber \\ \ &
 +pA_c + (1-p)A_p.
\end{align}

The maximum sum rate that the CR and PU can transmit, for a given $p$, can be
found out by solving the following optimization problem. The data rate
achievable by the CR is
\begin{align}
\max_{p_{fa},p_{md}}\ &  \bigg\lbrace -A_cpp_{md} - (1-p)(A_p-B_p-B_c)p_{fa}
+pA_c + (1-p)A_p\bigg\rbrace, \nonumber \\
\ & s.t, \mbox{ , } 0\leq p_{fa} \leq 1 \mbox{, } 0\leq p_{md} \leq 1.
\end{align}
It is easy to verify that the maximum is attained when $p_{fa}=p_{md}=0$. The
value of sum capacity will then be equal to that in the ideal case.

The inverse problem of the maximum sum rate is when we have a given $p_{fa}$
and $p_{md}$ and we want to see what value of $p$ maximizes the sum capacity.
\begin{align}
\max_{p}\ &  \bigg\lbrace -A_cpp_{md} - (1-p)(A_p-B_p-B_c)p_{fa}
+pA_c + (1-p)A_p\bigg\rbrace, \nonumber \\
\ & s.t \mbox{ , } 0\leq p \leq 1.
\end{align}
The objective function can be re-framed as
\begin{align}
\max_{p}\ &  \bigg\lbrace ((A_p - B_p - B_c)p_{fa} - A_cp_{md} + A_c - A_p)p -
(A_p - B_p - B_c)p_{fa} + A_p \bigg\rbrace, \nonumber \\
\ & s.t \mbox{ , } 0\leq p \leq 1.
\end{align}
The solution to this depends upon the term $M = \{(A_p - B_p - B_c)p_{fa} -
A_cp_{md} + A_c - A_p\}$ and is given by
\begin{align}
 p = \begin{cases}
     & 1, \mbox{ if } M>0 \\
     & 0, \mbox{ if } M<0 \\
     & \mbox{any value in } [0,1] \mbox{ if } M=0
     \end{cases}.
\end{align}
\section{Spectral Efficiency Factor in the Non Ideal Detection} \label{sec5}
In this section, we use expression of the sum capacity of CR
and PU to express spectral efficiency factor in terms of $p$, $p_{md}$
and $p_{fa}$. We look for the pairs $(p_{fa}, p_{md})$ which result in spectral
efficiency factor being more than unity for a given value of $p$ corresponding
to the case where the sum capacity in the CR
system is at least as much as was the capacity of PU alone without CR.
Then those $(p_{fa},
p_{md})$ pairs are explored which guarantee that the PU's performance loss is
within a limit. All such pairs make up admissible regions. Note that a
pair $(p_{fa},p_{md})$ that is not admissible is termed as inadmissible
and all such pairs constitute the inadmissible region.
Appropriately, we define the two types of regions (i) where the $\eta > 1$ -
called weakly admissible region and (ii) where the PU does not loose beyond a
limit - called strongly admissible region within a specified loss factor for
PU. We also find what pairs that are admissible if the PU cannot incur any loss
and term these as strong admissible pairs. We then discuss the effects of
$p_{fa}$, $p_{md}$ and $p$ and show that $p_{fa}$ has more profound
effect than $p_{md}$ on the performance of CR.

The spectral efficiency factor for the non ideal case is as shown below.
\begin{mythe}\label{thn}
The spectral efficiency factor in the non ideal case $\hat{\eta}$,
with average power constraints $P_p$ and $P_c$ for the PU and CR respectively is
given by
\begin{align}\label{nni}
\hat{\eta} =\ & \dfrac{C_c'+C_p'}{(1-p)C_p}
\end{align}
where $C_c'$ and $C_p'$ are specified in
Lemma \ref{rate1}.
\end{mythe}
\begin{proof}
Looking at definition \ref{speeff} and the definition of the
achievable data rates in the non ideal case from Lemma \ref{rate1} we have the
result.
\end{proof}
In the ideal case the spectral efficiency factor depended only on one parameter $p$, while in the non ideal case here it depends upon $p_{fa}$ and $p_{md}$ as well. Hence the following corollary helps in studying the rate at which spectral efficiency increases/decreases with change in any one of these parameters keeping the others fixed.
\begin{mycor}
 Spectral efficiency factor $\hat{n}$ increases with $p$ for a fixed $p_{md}$ with a rate given by 
 \begin{align}
  \dfrac{1}{(1-p)^2}\left(\dfrac{A_c(1-p_{md})}{A_p}\right).
 \end{align}
\end{mycor}
\begin{proof}
 By partially differentiating (\ref{nni}) w.r.t $p$ we get the result.
\end{proof}
\begin{mycor}
 Spectral efficiency factor $\hat{n}$ decreases with $p_{md}$ for a fixed $p$ with a rate given by 
 \begin{align}
  \dfrac{pA_c}{(1-p)A_p}.
 \end{align}
\end{mycor}
\begin{proof}
 By partially differentiating (\ref{nni}) w.r.t $p_{md}$ we get the result.
\end{proof}
\begin{mycor}
 Spectral efficiency factor $\hat{n}$ decreases with $p_{fa}$ with a rate given by 
 \begin{align}
  \dfrac{(A_p-B_p-B_c)}{A_p}.
 \end{align}
\end{mycor}
\begin{proof}
 By partially differentiating (\ref{nni}) w.r.t $p_{fa}$ we get the result.
\end{proof}
It is interesting to note that the rate of increase in $\hat{\eta}$ with $p$ is independent of $p_{fa}$. Also the rate of decrease in $\hat{\eta}$ with $p_{md}$ is independent of $p_{fa}$. Moreover, the rate of decrease in $\hat{\eta}$ with $p_{fa}$ is independent of both $p_{fa}$ and $p_{md}$. Hence, we can say that $p_{fa}$ decreases the spectral efficiency factor of the system equally at all occupancies.

Observe from Theorem \ref{thn} and Lemma \ref{rate1} that not all pairs of
$(p_{fa},p_{md})$, will result in the spectral efficiency factor to be greater
than one. This is because if there is high $p_{fa}$ then due to
interference both CR and PU will have reduced transmission rates. If there is
high $p_{md}$ then the CR will have less transmission rate. To study precisely
the effects of $(p_{fa},p_{md})$ and to characterize them in a region for which
$\hat{\eta}\geq1$ we define the notion of weak admissibility as follows
\begin{mydef}
A pair $(p_{fa},p_{md})$ is said to be weakly admissible for a Bernoulli
occupancy $p$ and average power constraints $P_p$ and $P_c$ on the PU and CR
respectively, if the spectral efficiency factor $\hat{\eta} \geq 1$.
\end{mydef}
When the CR interferes in the PU communication, there will be a loss to the
PU. Suppose we are given a limit beyond which the PU cannot incur a loss. It
is of interest to see what values of $(p_{fa},p_{md})$ can ensure that the loss
to the PU is below this limit. We model this in the form of the loss
factor. The loss factor of the PU is the fraction of data rate it looses due to
the intervention of the CR in its communication. Formally,
\begin{mydef} \label{gam}
The loss factor of the PU is
\begin{align}
\gamma  = \dfrac{C_p'}{(1-p)C_p} .\nonumber
\end{align}
\end{mydef}
Note that $0< \gamma \leq 1$. We now wish to see what $(p_{fa},p_{md})$ will
guarantee a loss below $\gamma $. This gives rise to notion of strong
admissibility with a loss factor $\gamma$ as follows,
\begin{mydef}
A pair $(p_{fa},p_{md})$ is said to be strongly admissible with a loss factor
$\gamma$ for a Bernoulli occupancy $p$ and average power constraints $P_p$ and
$P_c$ on the PU and CR respectively, if $C_p'\geq \gamma(1-p)C_p$. In
particular if $\gamma = 1$ we say that $(p_{fa},p_{md})$ are strongly
admissible.
\end{mydef}
Now, we characterize the admissible regions based on the above definitions.
\subsection{Characterization of  Admissible Regions} \label{num}
In this subsection we characterize three types of admissible regions viz
(i) weakly admissible region - Theorem \ref{weakk}, (ii) strongly admissible
region with a loss factor $\gamma$ - Theorem \ref{los} and (iii) strongly
admissible region - Theorem \ref{strr}.
\begin{mythe} \label{weakk}
The pair $(p_{fa},p_{md})$ are weakly admissible if,
\begin{enumerate}[(1)]
 \item $0\leq p_{fa} \leq 1$.
 %\item $0\leq p_{md} \leq 1$.
 \item $p_{md}\leq 1 -\left(\dfrac{1-p}{p}\right)\left(\dfrac{A_p - B_p -
B_c}{A_c}\right)p_{fa}$.
\end{enumerate}
\end{mythe}
\begin{proof}
The first constraint is kept there to realize that the other constraint
can be satisfied for almost all values of $(p_{fa},p_{md})$ if $p_{fa}$ is not
confined to the interval of $[0,1]$.
From definition of weak admissibility
\begin{align}
 \ & \hat{\eta}\geq 1 \nonumber \mbox{ or,} \\
 \ & (1-p)((1-p_{fa})A_p+p_{fa}B_p)) + p(1-p_{md})A_c + (1-p)p_{fa}B_c \geq
(1-p)A_p
\end{align}
Grouping terms of $p_{fa}$ and $p_{md}$ together, we have
\begin{align}
\ & (1-p)[A_p - B_p - B_c]p_{fa} + pA_cp_{md} \leq pA_c \nonumber \mbox{ or,} \\
\ & p_{md}\leq 1 -\left(\dfrac{1-p}{p}\right)\left(\dfrac{A_p - B_p -
B_c}{A_c}\right)p_{fa}. \nonumber
\end{align}
\end{proof}
Observe that the region for strong admissibility with a factor
$\gamma$, only depends on $p_{fa}$. This is because strong admissibility is
concerned with the data rate of PU only. For a $p_{fa}$ that is strongly
admissible with a factor $\gamma$ all values of $p_{md}$ are admissible.
Therefore we plot $p_{fa}$ for strong admissibility. Also we provide
characterization of strong admissible regions with a loss factor $\gamma$ in
terms of $p_{fa}$ only as follows
\begin{mythe}\label{los}
The value of $p_{fa}$ is strongly admissible for a given loss factor $\gamma$ of
the PU if,
$$
0 \leq p_{fa} \leq
\min\bigg\lbrace{1,\dfrac{A_p(1-\gamma)}{A_p-B_p}}\bigg\rbrace.
$$
\end{mythe}
\begin{proof}
We have, $\gamma \leq 1$. From (1), Lemma \ref{rate1} and definition \ref{gam}
we have
\begin{align}
\ & (1-p_{fa})A_p + p_{fa}B_p \leq \gamma A_p \nonumber \mbox{ or,} \\
\ & p_{fa} \leq \dfrac{A_p(1-\gamma)}{A_p-B_p}.
\end{align}
Now since, $p_{fa}\leq 1$, the result follows.
\end{proof}
\begin{mythe}\label{strr}
Strongly admissible region is characterized by $p_{fa}=0$.
\end{mythe}
\begin{proof}
For strong admissibility we have $\gamma =1$. Thus,
\begin{align}
 (1-p)((1-p_{fa})A_p+p_{fa}B_p)) = (1-p)C_p \nonumber \\
\end{align}
Note that, $C_p=A_p$, hence we have
\begin{align}
 p_{fa}(A_p - B_p) = 0.
\end{align}
$A_p = B_p$ in the limit would imply that $P_s \ll \sigma^2$ which is never the
case. Hence, $p_{fa}=0$.
\end{proof}
Theorem \ref{strr} implies that any detection technique with $p_{fa}>0$ will
lead to spectral efficiency loss for the PU. Since all the practical techniques,
will have $p_{fa}>0$, we have introduced the notion of strong admissibility
with a loss factor.

Note that the weak admissible region occupies a non zero area in the $(p_{fa},p_{md})$
plane for $p>0$ as evident from Theorem \ref{weakk}. An instance of the weak
and strong admissible region is shown in Fig. \ref{weak}. The set of
strongly admissible $(p_{fa},p_{md})$ pairs (i.e. $\gamma =1$) is a sub-region
of the weakly admissible pairs $(p_{fa},p_{md})$ as shown in Fig. \ref{weak}.
\subsection{Strong Admissible $p_{fa}$ with a loss factor $\gamma$.}
A typical curve of the strongly admissible $p_{fa}$  with a loss factor
$\gamma$ against $1-\gamma$ is shown in Fig. \ref{loss}. As we decrease
$\gamma$, $1-\gamma$ increases the rate at which the admissible value of
$p_{fa}$ increases is $\dfrac{A_p}{A_p- B_p}$.
\begin{mythe}
The range of strongly admissible values of $p_{fa}$ for a loss factor $\gamma$
is independent of the PU channel occupancy. Furthermore, for a loss factor of
$\gamma \leq \dfrac{B_p}{A_p}$ the range of values is
the entire interval $[0,1]$.
\end{mythe}
\begin{proof}
From Theorem \ref{los} it is clear that admissible $p_{fa}$ values for a given
$\gamma$ is not dependent on $p$. The second statement of the Theorem, follows
from the Theorem \ref{los} substituting $p_{fa}=1$ in the second
inequality.
\end{proof}
We call $\gamma = \dfrac{B_p}{A_p}$ as the full admissible point. Observe that
for $\gamma = 1$, i.e. the strong admissible case $p_{fa}=0$. Also note that the
admissible values of $p_{fa}$ are highly dependent on the relative power levels
of the PU and CR ($RS$) transmission. As the value $RS$ increases the values
of
strongly admissible $p_{fa}$ for a given loss factor also increases as shown in
Fig. \ref{pflos}. Similarly as %As the PU to CR
$RS$ increases the value of the full admissible point increases as illustrated
in Fig. \ref{bptab}.

Next we consider the weak admissible region.  Fig. \ref{exp4} shows the spectral efficiency factor plot
for a particular value of $p$ and $RS$ with the admissible and inadmissible
regions marked. Fig. \ref{spec} shows the spectral
efficiency factor for various values of $p$ for a constant $RS$. Observe that
the admissible region increase as $p$ increases. This implies that for channel
with low primary occupancy a lousy detector (with high $p_{fa}, p_{md}$) will
also result in increase in spectral efficiency.   Similarly for $p_{fa}=0$ the
PU will transmit its original capacity without CR and hence any value of
$p_{md}$ would be weakly admissible. Also observe from Fig. \ref{snrvar} that
when $p$ is fixed and $RS$ is increased, the admissible region increases. This
is justified because, when the $RS$ is more the CR transmits relatively lesser
power which causes lesser interference when compared to when $RS$ is lower.
This translates into higher values of admissible $p_{fa}$ for higher $RS$ when
compared to the lower $RS$.

We show the weakly admissible and inadmissible regions in the Figs. \ref{thrd}
and \ref{pffm}.

\section{Comparison of Various Detection Techniques in light of the Admissible
Regions}\label{sec6}
We compare few of the standard detection techniques viz. energy detection (ed),
matched filter (mf) and magnitude squared coherence (msc).

For every detection technique (dt) we have a relation $f_{dt}$ between the
$p_{md}$ and $p_{fa}$, i.e. $p_{md}= f_{dt}(p_{fa})$ for a given average
received
signal to noise ratio (SNR). The function $f_{dt}$ for all standard detection
techniques is monotone non increasing. For a given $p_{fa}$ we wish to see
whether any value of $p_{md}$ is admissible.

Denote $h(p_{fa}) = 1 - \left(\dfrac{1-p}{p}\right)\left(\dfrac{A_p - B_p
-B_c}{A_c}\right)p_{fa}$. We then wish to find for a detection technique whether
for a given $p_{fa}$, $f_{dt}(p_{fa}) \leq h(p_{fa})$ holds. If yes, then we say
that the detection technique operating at that $(p_{fa},p_{md})$ pair is
admissible and consequently all such points on the curve of detection technique
are referred to as admissible region of the detection technique. Note that,
this admissible region of a detection technique is dependent on average
received SNR.

We provide the functions $f_{dt}$ for three detection techniques
mentioned above. In what follows the value of $p_{md}$ is a function of $p_{fa}$ for all other parameters fixed.
\begin{align}
\ & p_{mded} = f_{ed}(p_{faed}) =
Q_{\chi^2}\left(2P^{-1}\left(1-p_{faed},L,\dfrac{MLP_p}{2\sigma^2}\right)\right), \mbox{ for a fixed $M$,$L$ and $P_p$,} \nonumber
\\
\ & p_{mdmf} = f_{mf}(p_{famf}) = 1- Q\left(Q^{-1}(p_{famf}) -
\dfrac{E}{\sigma\sqrt{E}}\right), \mbox{ for a fixed $E$ and $\sigma$,} \nonumber \\
\ & p_{mdmsc} = f_{msc}(p_{famsc}) = P_{CDF}\left(\left(1
- p_{famsc}^{(\frac{1}{L-1})}\right)|L,|\gamma|^2\right),\mbox{ for a fixed $L$ and $\gamma$.}
\end{align}
where, $Q_{\chi^2}(x,v,\delta)$ is the non central $\chi^2$ distribution of $x$
with $v$ degrees of freedom and positive non centrality parameter $\delta$ and
$P^{-1}(x,a)$ is the inverse lower incomplete $\Gamma$ function. $Q(x)$ is the
standard Q-function with $Q^{-1}$ being the inverse $Q$ function.
$P_{CDF}(|\hat{\gamma}|^2|L,|\gamma|^2)$ is as given in \cite{cart}. Here, $L$
denotes the number of disjoint sequences that an original sequence of length $N$
is divided into each of length $M$ so that $M=\dfrac{N}{L}$. $E$ is
the energy of the transmitted signal. $|\gamma|^2$ is the magnitude squared
coherence which is the magnitude squared of the spectral coherence. For more
details, the reader is referred to \cite{murthy} and the references therein. 

We now plot $p_{md}$ vs. $p_{fa}$ and look for those areas where the curve of
detection region lies inside the weak admissible region. For all points lying on
that part of the curve the detection technique is admissible.

Figs. \ref{snrdet24}, \ref{snrdet26} show that, for low SNR regimes, all
values of $(p_{fa},p_{md})$ offered by the matched filter detector are
admissible. For the magnitude squared coherence detector  majority of the
points on the curve are admissible except those at higher end of
$p_{fa}$ values. For the energy detector many points on the curve lie outside
the admissible region. Hence, using energy detector, one has to be 
careful about the operating point $(p_{fa},p_{md})$. It follows that the matched filter
performs better than magnitude squared coherence detector which performs better
than the energy detector. These observations are in line with those in \cite{wang}.
\section{Conclusion} \label{sec7}
In this paper we have capacity region of a point to point CR channel. For ideal
detection of PU occupancy we have seen the limits on rates. We have defined
spectral efficiency factor and studied its variations with respect to the
occupancy probability for different relative power levels of the PU and CR.
Then we built the case where the CR performs non ideal detection of the PU
presence and specified capacity region for the same. With expressions derived
we brought out the limits, under different occupancies and relative PU and CR
power levels, on the false alarm and missed detection for which the spectral
efficiency of the overall system increases and also for which the PU is not at a
loss greater than a specified value. We have discussed the effects of false
alarm, missed detection and channel occupancy with respect to spectral
efficiency factor. Finally we compared our result with standard detection
techniques viz. energy detection, matched filter and magnitude squared coherence
and found that the matched filter performs best followed by magnitude squared
coherence followed by energy detection with respect to spectral efficiency
factor.
\section*{Acknowledgment}
The work was sponsored by Department of Information Technology, Government of
India, under the grant for the project ``High Performance Cognitive Radio
Networks at IIT Bombay and IIT Hyderabad''.

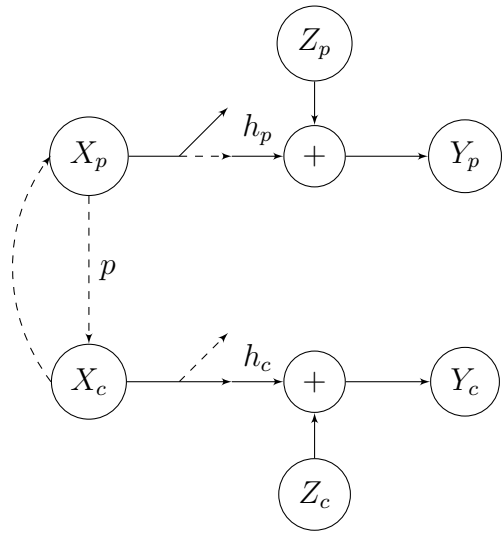
\begin{figure}[H]
\centering
\tikzstyle{channel} = [draw,fill=blue!20,minimum size=3em]
\tikzstyle{line} = [draw, -latex']
\tikzstyle{par} = [draw, parabola -latex']
\def\soledge[#1,#2,#3,#4];{
        \draw[dashed,-latex,#4] (#1) -- #3 (#2);}
\begin{tikzpicture}[>=latex']
    \node (PRT) at (0,3) [circle,draw,label=above:] {$X_p$};
    \node (CRT) at (0,0) [circle,draw,label=below:] {$X_c$};
    \node (PUR) at (5,3) [circle,draw,label=above:] {$Y_p$};
    \node (CRR) at (5,0) [circle,draw,label=below:] {$Y_c$};
    \node (PUN) at (3,4.5) [circle,draw,label=above:] {$Z_p$};
    \node (CRN) at (3,-1.5) [circle,draw,label=below:] {$Z_c$};
    \node (PUA) at (3,3) [circle,draw,label=above:] {$+$};
    \node (CRA) at (3,0) [circle,draw,label=below:] {$+$};
    \draw [style = dashed, ->](CRT.west) arc (-143:-217:2.5cm);
    \path [line] (1.9,3) to node [auto] {$h_p$} (PUA.west);
    \path [line] (1.9,0) to node [auto] {$h_c$} (CRA.west);
    \path [line] (1.2,3) to (1.85,3.65) node [auto] {} ;
    \path [line] (1.2,0) to (1.9,0) node [auto] {} ;
    \path [line, dashed] (1.2,0) to (1.85,0.65) node [auto] {} ;
    \path [line, dashed] (1.2,3) to (1.9,3) node [auto] {} ;
    \path [line,dashed] (PRT.south) to node [auto] {$p$}(CRT.north);
    \draw [-] (PRT.east) to (1.2,3) node [auto]{};
    \draw [-] (CRT.east) to (1.2,0) node [auto]{};
    \draw [->] (PUN.south) to node [auto] {} (PUA.north);
    \draw [->] (CRN.north) to node [auto] {} (CRA.south);
    \draw [->] (PUA.east) to node [auto] {} (PUR.west);
    \draw [->] (CRA.east) to node [auto] {} (CRR.west);
\end{tikzpicture}
\caption{The Cognitive radio channel when CR ideally detects the PUs presence.}
\label{idcha}
\end{figure}
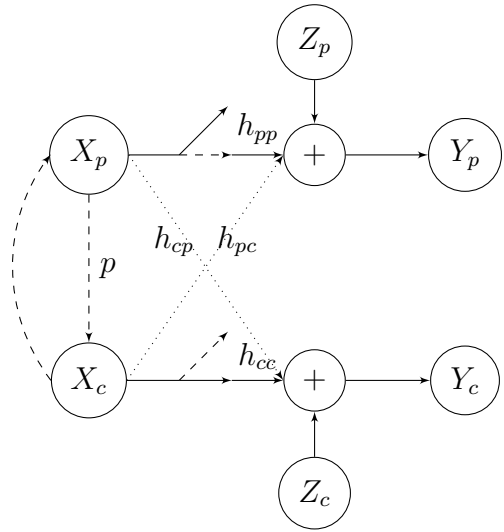
\begin{figure}[H]
\centering
\tikzstyle{channel} = [draw,fill=blue!20,minimum size=3em]
\tikzstyle{line} = [draw, -latex']
\tikzstyle{par} = [draw, parabola -latex']
\def\soledge[#1,#2,#3,#4];{
        \draw[dashed,-latex,#4] (#1) -- #3 (#2);}
\begin{tikzpicture}[>=latex']
    \node (PRT) at (0,3) [circle,draw,label=above:] {$X_p$};
    \node (CRT) at (0,0) [circle,draw,label=below:] {$X_c$};
    \node (PUR) at (5,3) [circle,draw,label=above:] {$Y_p$};
    \node (CRR) at (5,0) [circle,draw,label=below:] {$Y_c$};
    \node (PUN) at (3,4.5) [circle,draw,label=above:] {$Z_p$};
    \node (CRN) at (3,-1.5) [circle,draw,label=below:] {$Z_c$};
    \node (PUA) at (3,3) [circle,draw,label=above:] {$+$};
    \node (CRA) at (3,0) [circle,draw,label=below:] {$+$};
    \draw [style = dashed, ->](CRT.west) arc (-143:-217:2.5cm);
    \path [line] (1.9,3) to node [auto] {$h_{pp}$} (PUA.west);
    \path [line] (1.9,0) to node [auto] {$h_{cc}$} (CRA.west);
    \path [line] (1.2,3) to (1.85,3.65) node [right] {} ;
    \path [line] (1.2,0) to (1.9,0) node [auto] {} ;
    \path [line, dashed] (1.2,0) to (1.85,0.65) node [auto] {} ;
    \path [line, dashed] (1.2,3) to (1.9,3) node [auto] {} ;
    \path [line,dashed] (PRT.south) to node [auto] {$p$} (CRT.north);
    \draw [-] (PRT.east) to (1.2,3) node [auto]{};
    \draw [-] (CRT.east) to (1.2,0) node [auto]{};
    \draw [->,dotted] (PRT.east) to node [auto] {$h_{pc}$} (CRA.west);
    \draw [->,dotted] (CRT.east) to node [auto] {$h_{cp}$} (PUA.west);;
    \draw [->] (PUN.south) to node [auto] {} (PUA.north);
    \draw [->] (CRN.north) to node [auto] {} (CRA.south);
    \draw [->] (PUA.east) to node [auto] {} (PUR.west);
    \draw [->] (CRA.east) to node [auto] {} (CRR.west);
\end{tikzpicture}
\caption{The Cognitive radio channel when CR non ideally detects ideally detects
the PUs presence. The dashed instance shows when PU transmits alone and dotted
when PU and CR transmit together and solid when CR alone transmit.}
\label{nidcha}
\end{figure}
\begin{figure}[H]
 \centering
 \def\firstcircle{(0,0) circle (1.5cm)}
\def\thirdcircle{(0:2cm) circle (1.5cm)}
\begin{tikzpicture}
    \draw \firstcircle node[below] {};
    \draw \thirdcircle node [below] {};
    \draw[->](-0.8,-1) -- (-3.5,-2.2)node[below, text width = 5cm]{Spectral hole detected at CR Txt but not at CR Rx, $p_{fa}\neq0$ };
    \draw[->](2.8,-1) --(5.5,-2.2) node[below, text width = 5cm]{Spectral hole detected at CR Rx  but not at CR Txt, $p_{md}\neq0$};
    \draw[->](1,0) -- (1,2.2)node[above,text width = 10cm]{Spectral hole detected at CR Txt and Rx both the case of single switch model, $p_{fa} = p_{md} = 0$};
    \draw[-](2.5,-.5) -- node[above]{}(2.5,-.5);
    \draw[-](-0.5,-0.5) -- node[above]{}(-0.5,-.5);
    \begin{scope}
      \clip \firstcircle;
      \fill[blue,opacity = 0.5] \thirdcircle;
    \end{scope}
\end{tikzpicture}
\caption{The two switch model shown as the special case of the non ideal case system model.}
\label{gen1}
\end{figure}
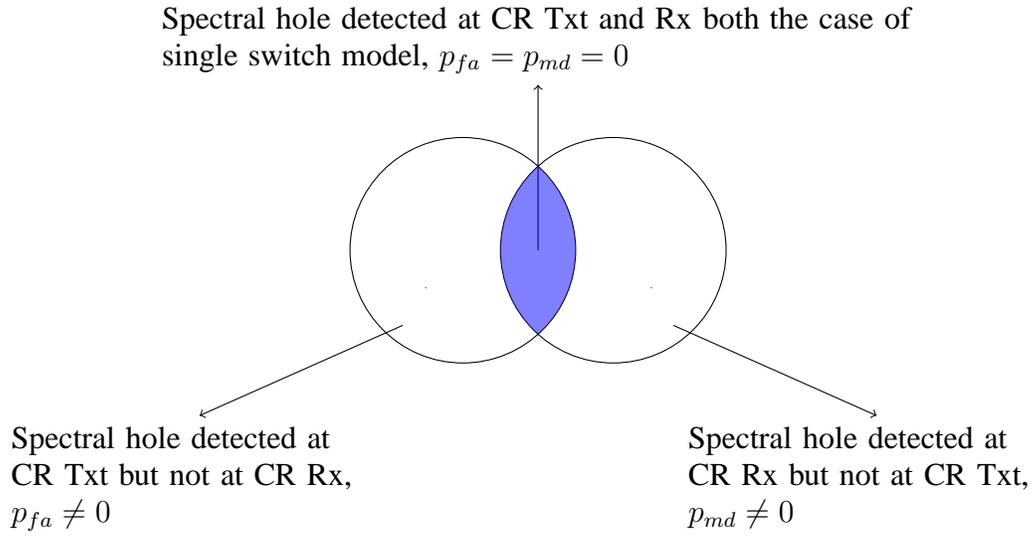
\begin{figure}[H]
 \centering
\begin{tikzpicture}
\tikzstyle{line} = [draw, -latex']
    \draw  [->] (-1,0) -- coordinate (x axis mid) (4,0) node[right]{$R_c$};
    \draw [->] (0,-1) -- coordinate (y axis mid) (0,4) node[above]{$R_p$};
    \draw [-](3,0)node[below]{$C_c$}--(0,3) node[left]{$C_p$};
\end{tikzpicture}
\caption{Rate region of  ideal CR.}
\label{idrate}
\end{figure}
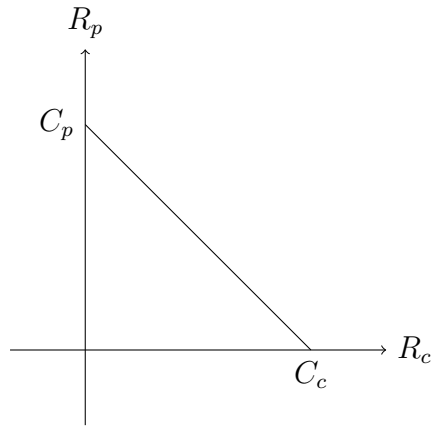
\begin{figure}[H]
\centering
\includegraphics[scale=0.6]{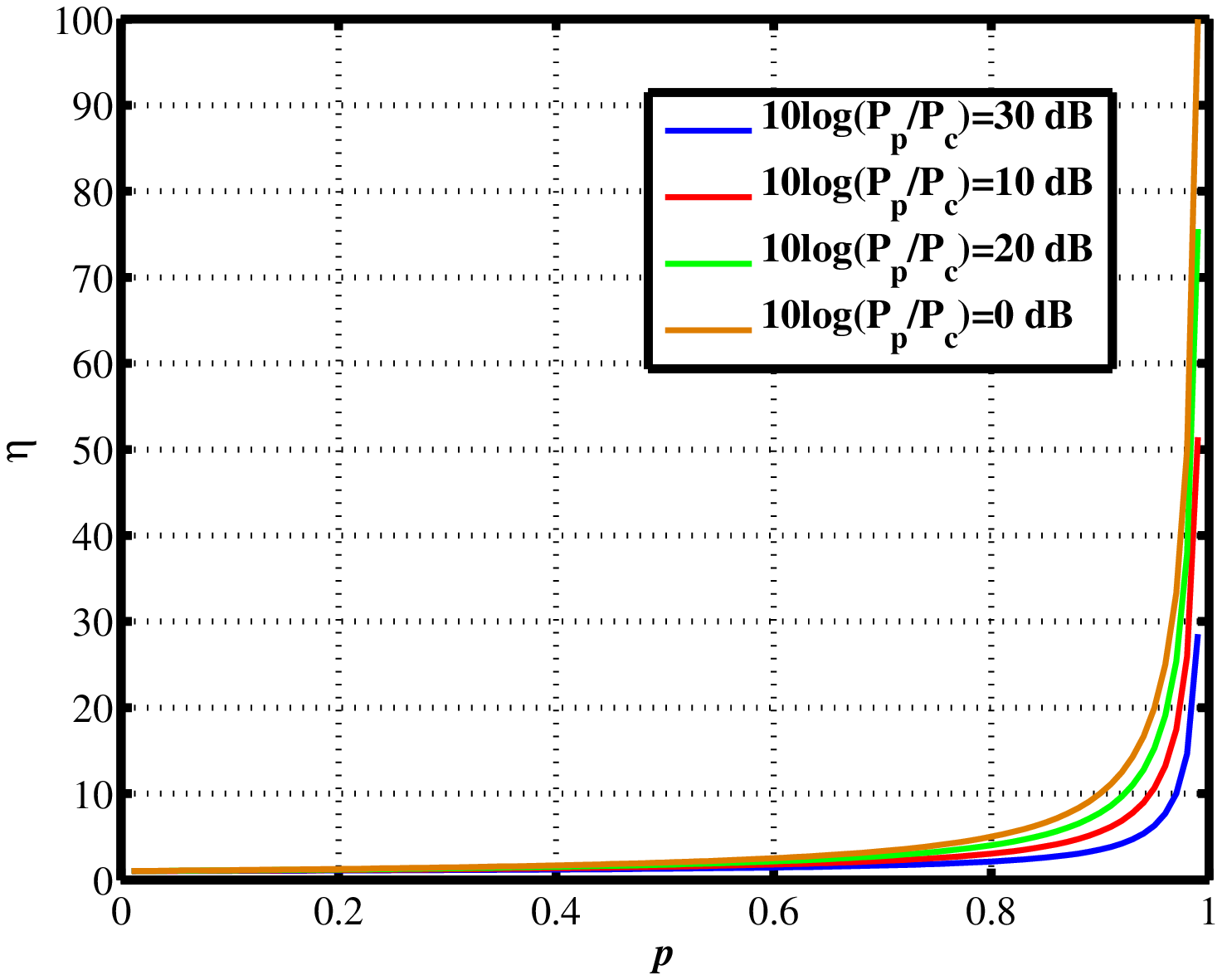}
\caption{Dependence of $\eta$ on $p$ for various relative power levels of PU
and CR.}
\label{exm}
\end{figure}
\begin{figure}[H]
\centering
\includegraphics[scale=0.6]{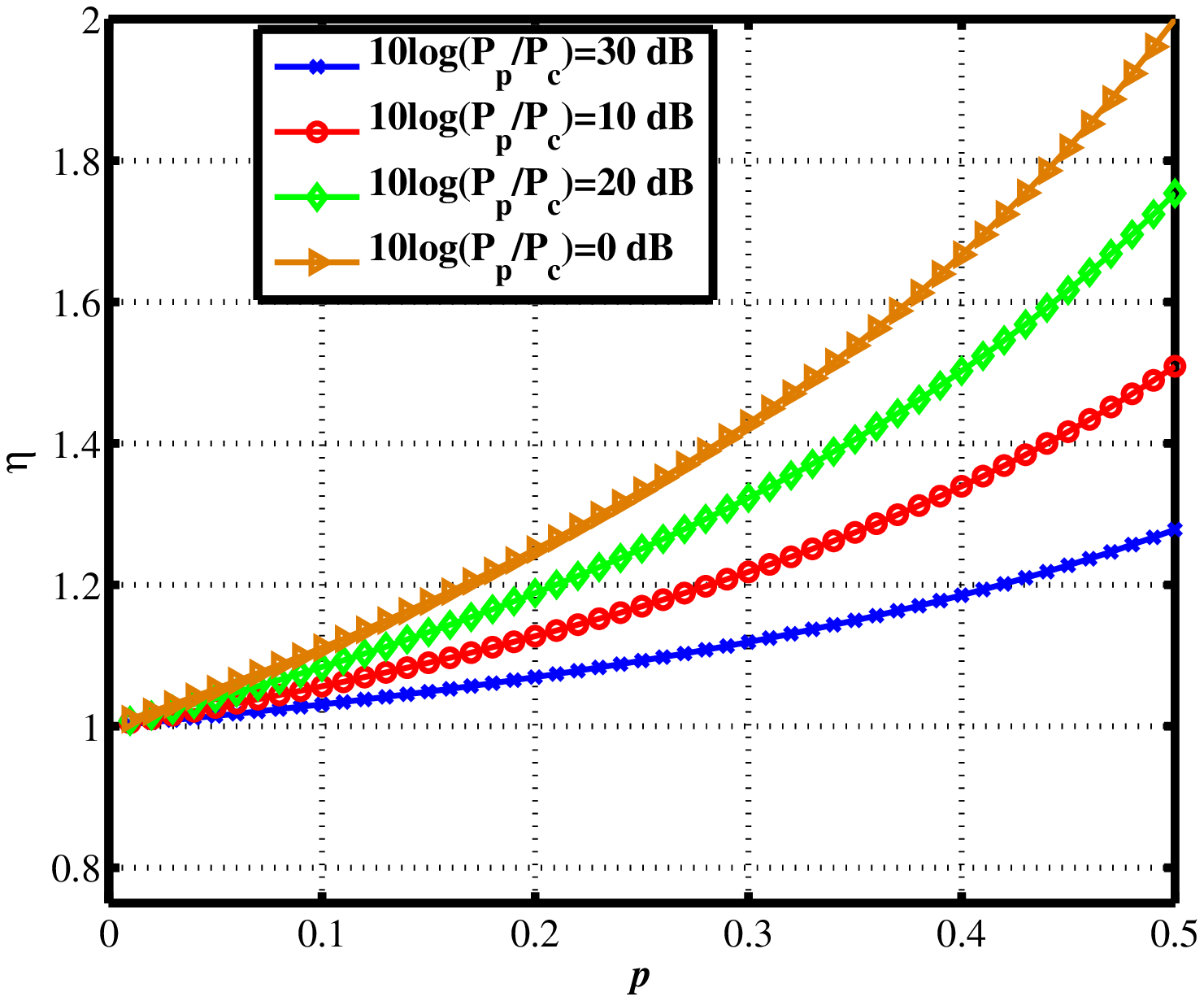}
\caption{Dependence of $\eta$ on small values $p$ for various relative
power levels of PU and CR.}
\label{exm1}
\end{figure}
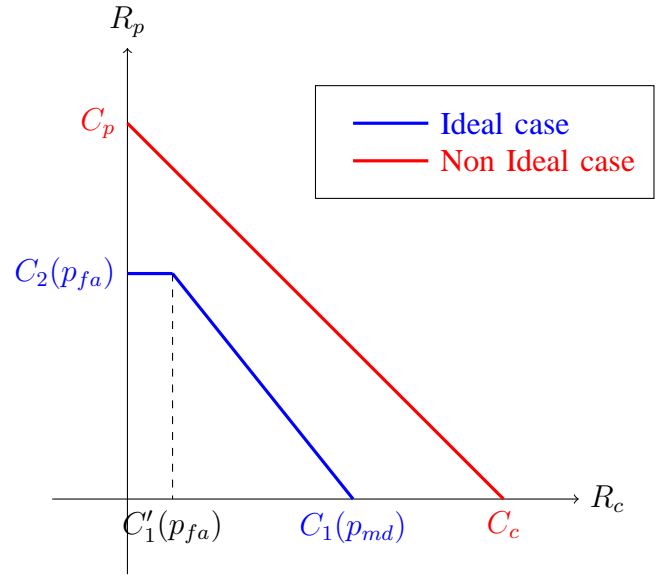
\begin{figure}[H]
\centering
\begin{tikzpicture}
\tikzstyle{line} = [draw, -latex']
\draw  [->] (-1,0) -- coordinate (x axis mid) (6,0) node[right]{$R_c$};
\draw [->] (0,-1) -- coordinate (y axis mid) (0,6) node[above]{$R_p$};
\draw [-, blue, very thick](3,0)node[below]{$C_1(p_{md})$}--(0.6,3)
node[left]{};
\draw [-] (2.5,4) -- (7,4) -- (7,5.5) -- (2.5,5.5) -- (2.5,4);
\draw [-,blue, very thick](0.6,3)node[below]{}--(0,3) node[left]{$C_2(p_{fa})$};
\draw [-,blue, very thick](3,5)node[below]{}--(4,5) node[right]{Ideal case};
\draw [-,red, very thick](3,4.5)node[below]{}--(4,4.5) node[right]{Non Ideal
case};
\draw [-, dashed](0.6,3) -- (0.6,0) node[below]{$C_1'(p_{fa})$};
\draw [-, red, very thick](5,0)node[below]{$C_c$}--(0,5) node[left]{$C_p$};
\end{tikzpicture}
\caption{Typical rate region of non ideal case.}
\label{nidrate}
\end{figure}
\begin{figure}[H]
 \centering
\includegraphics[scale=0.3]{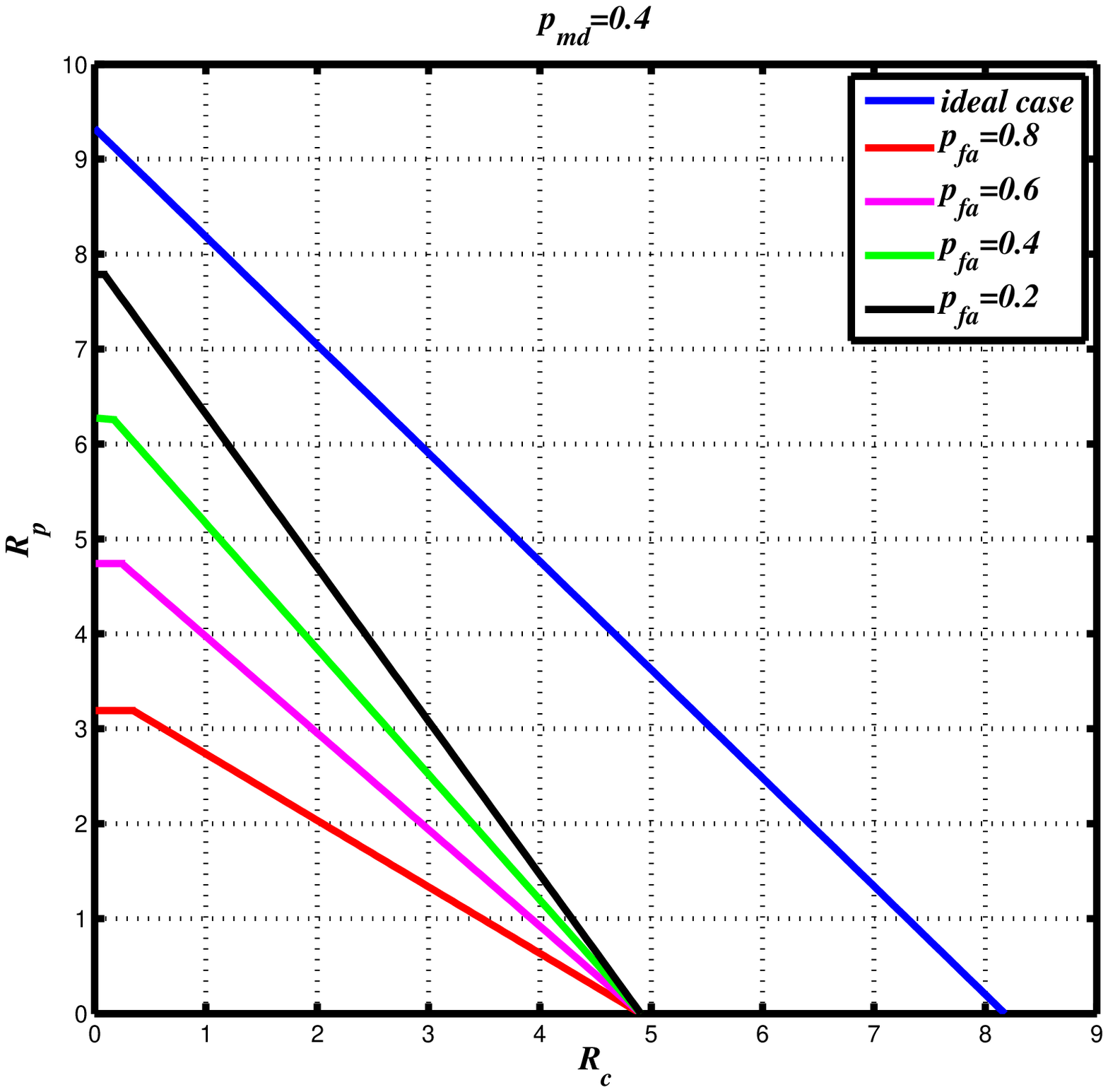}
\caption{Rate region of non ideal case compared with non ideal for fixed
$p_{md}$ but variable $p_{fa}$.}
\label{pfrate}
\end{figure}
\begin{figure}[H]
 \centering
\includegraphics[scale=0.3]{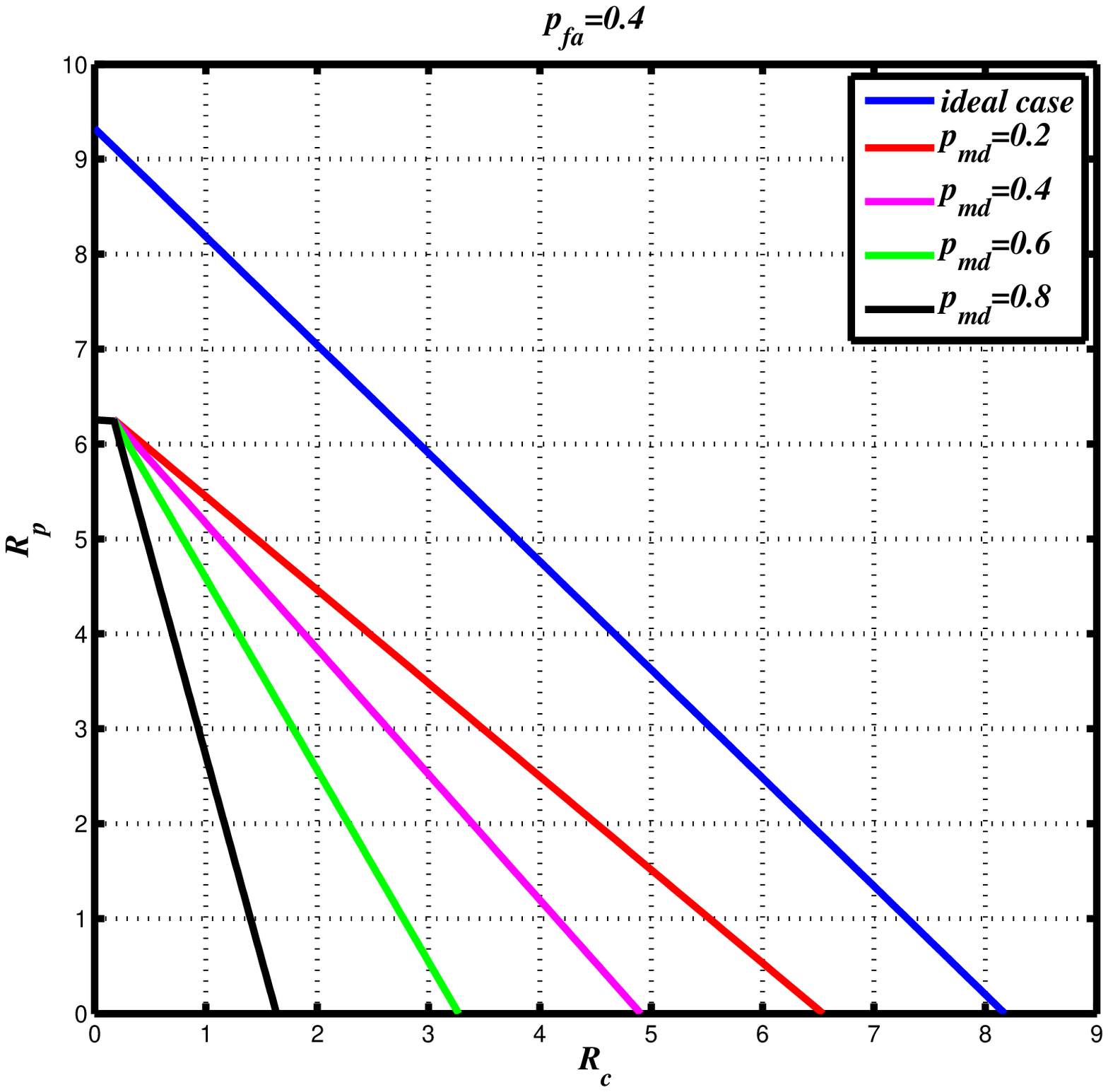}
\caption{Rate region of non ideal case compared with non ideal for fixed
$p_{fa}$ but variable $p_{md}$.}
\label{pmrate}
\end{figure}
\begin{figure}[H]
\centering
\includegraphics[scale=0.3]{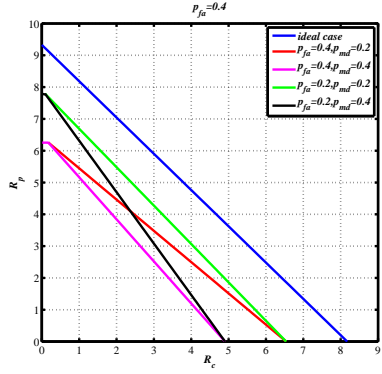}
\caption{Rate region of non ideal case compared with non ideal for two
$p_{fa}$ and $p_{md}$ values. When $p_{fa}$ and $p_{md}$ both decrease the
rate region becomes closer to the ideal case.}
\label{ratepmpf}
\end{figure}
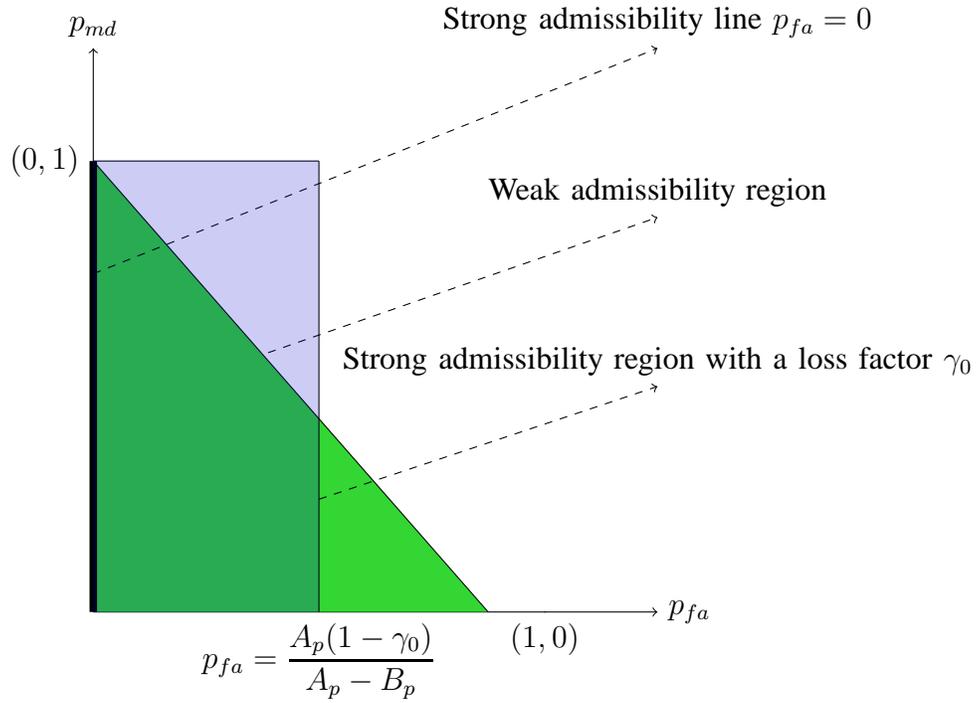
\begin{figure}[H]
\centering
\begin{tikzpicture}[scale=1.5]
\tikzstyle{line} = [draw, -latex']
\draw [->] (0,0) -- coordinate (x axis mid) (5,0)
node[right]{$p_{fa}$};
\draw [->] (0,0) -- coordinate (y axis mid) (0,5) node[above]{$p_{md}$};
\fill[color=green!80!black,opacity=0.8](0,0)--(0,4)--(3.5,0)--cycle;
\draw [-](0,4) -- (3.5,0)node[right]{};
\draw [-](0,4) -- (2,4)node[right]{};
\draw [-](2,4) -- (2,0)node[right]{};
%\draw [-, dashed](4,0) -- (4,4)node[right]{$(1,1)$};
\draw [-, dashed](4,0) -- (4,0)node[below]{$(1,0)$};
\draw [-,line width=1mm] (0,0) -- (0,4) node[left]{$(0,1)$};
\draw [->, dashed](0,3) -- (5,5)node[above]{Strong admissibility line
$p_{fa}=0$};
\draw [->, dashed](1.55,2.3) -- (5,3.5)node[above]{Weak admissibility region};
\draw [->, dashed](2,1) -- (5,2)node[above]{Strong admissibility region with a
loss factor $\gamma_0$};
\fill[color=blue!80!black,opacity=0.2](0,0)--(0,4)--(2,4)--(2,0)node[below,
color=black,opacity=1]{ $p_{fa} = \dfrac {A_p(1-\gamma_0)}{A_p-B_p}$} --(0 , 0);
\end{tikzpicture}
\caption{Weak admissible region appears as shown above. It increases with
increase in value of $p$ and when PU to CR power ratio increases. Note here the
value of $p_{fa}$ which causes the change in admissible region. This is
explained later. Strong admissible region for all values of $p$ and CR and PU
powers transmit is the line $p_{fa} =0$. Strong admissible region with a loss
factor $\gamma_0$ is the shaded rectangle. As $\gamma_0$ increases the size of
the rectangle decreases and eventually becomes the line $p_{fa}=0$ for
$\gamma_0=1$. Note that strong admissible with a loss factor $\gamma_0$ is not
necessarily a sub region of the weak admissible region.}
\label{weak}
\end{figure}
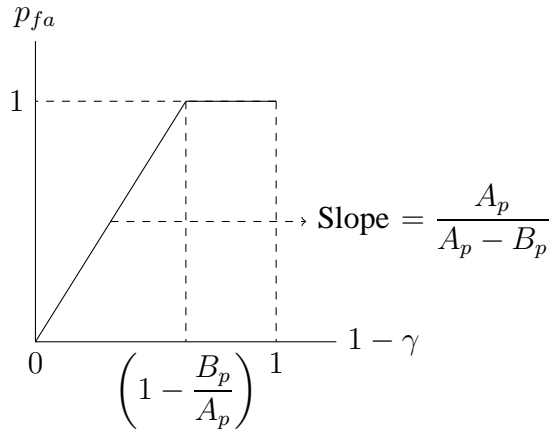
\begin{figure}[H]
\centering
\begin{tikzpicture}[scale=0.8]
\tikzstyle{line} = [draw, -latex']
\draw  [-] (0,0) -- coordinate (x axis mid) (5,0)
node[right]{$1-\gamma$};
\draw [-] (0,0) -- (0,0) node[below]{$0$};
\draw [-] (0,0) -- coordinate (y axis mid) (0,5) node[above]{$p_{fa}$};
\draw [-] (0,0) -- (2.5,4) node[above]{};
\draw [-](2.5,4) -- (4,4)node[right]{};
\draw [->,dashed](1.25,2) -- (4.5,2)node[right]{Slope $ =
\dfrac{A_p}{A_p-B_p}$};
\draw [-,dashed](2.5,4) --
(2.5,0)node[below]{$\left(1-\dfrac{B_p}{A_p}\right)$};
\draw [-, dashed](4,4) -- (0,4)node[left]{$1$};
\draw [-, dashed](4,4) -- (4,0)node[below]{$1$};
\end{tikzpicture}
\caption{$p_{fa}$ vs $(1-\gamma)$. As $(1-\gamma)$ increases, the
maximum value of $p_{fa}$ increases up until $\gamma = \dfrac{B_p}{A_p}$, where
$p_{fa}\leq 1$ becomes admissible.}
\label{loss}
\end{figure}
\begin{figure}[H]
 \centering
\includegraphics[scale=0.5]{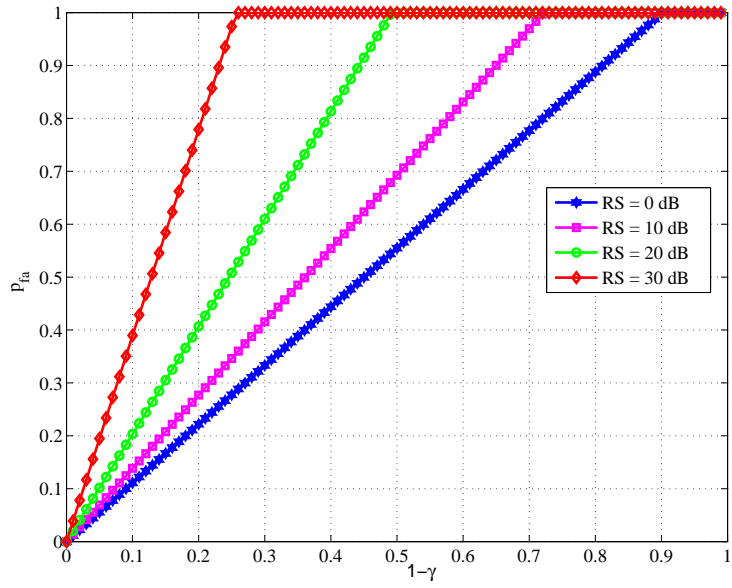}
\caption{$p_{fa}$ vs $(1-\gamma)$ for different values of relative power
levels of CR and PU; $RS$ and $p=0.4$.}
\label{pflos}
\end{figure}
\begin{figure}[H]
 \centering
\includegraphics[scale=0.5]{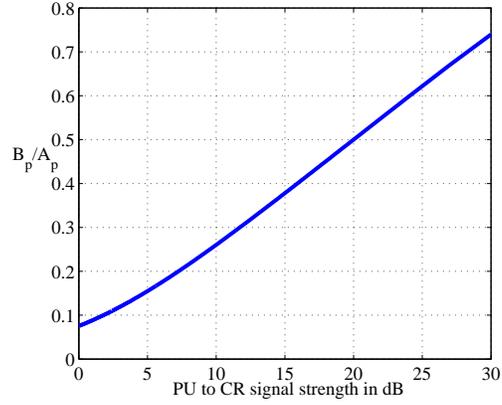}
\caption{Plot of full admissible point vs relative power levels of the CR and
PU; $RS$. for $p=0.4$}
\label{bptab}
\end{figure}
\begin{figure}[H]
 \centering
\includegraphics[scale=0.4]{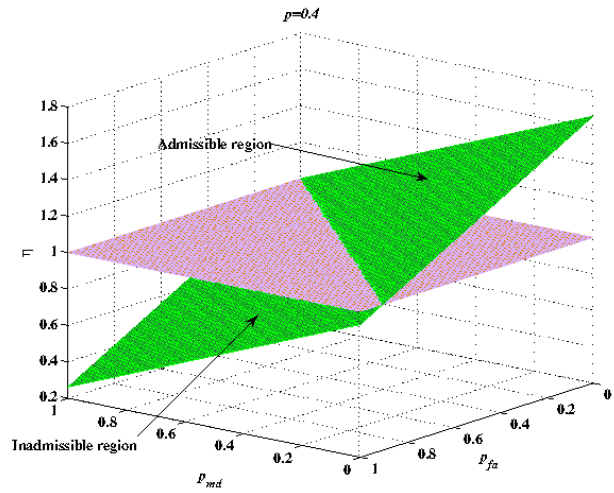}
\caption{Spectral efficiency factor variation vs $p_{fa}$ and
$p_{md}$ for $p=0.4$ and $RS=0$ dB showing  weakly admissible and inadmissible
regions.}
\label{exp4}
\end{figure}
\begin{figure}[H]
 \centering
\includegraphics[scale=0.4]{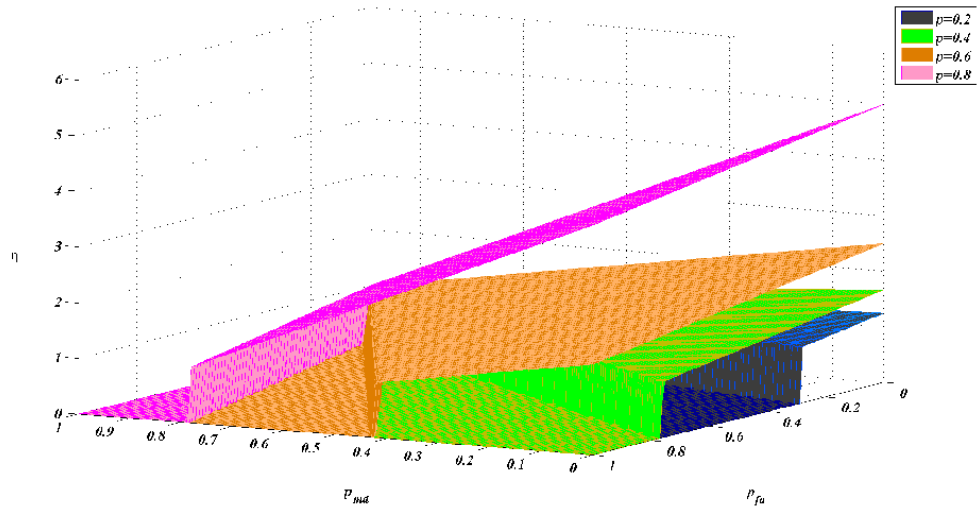}
\caption{Spectral efficiency factor variation vs $p_{fa}$ and
$p_{md}$ for various values of $p$ and $RS=0$ dB, highlighting the change in
admissible region with $p$. A cut has been introduced to separately depict the
spectral efficiency factor in the weakly admissible regions only.}
\label{spec}
\end{figure}
\begin{figure}[H]
\centering%
\includegraphics[scale=0.35]{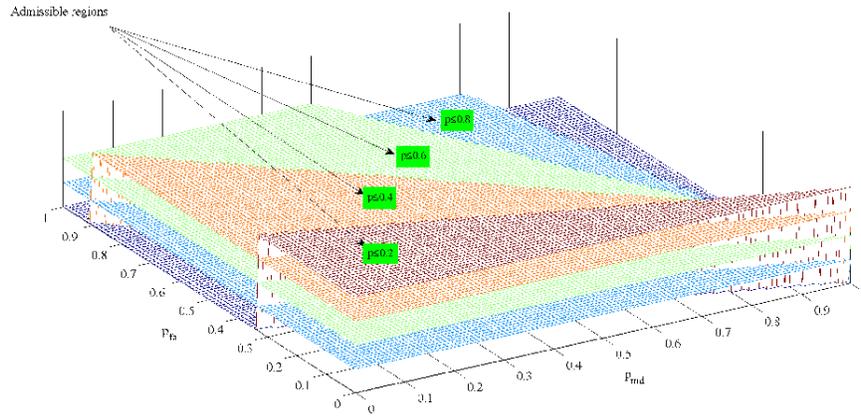}
\caption{$p_{fa}$ vs $p_{md}$ plot showing weakly admissible regions. Note that
the weakly admissible region for $p=p_0$ (say 0.6) is the sum of shaded weakly
admissible regions for all $p$ less than $p_0$. So the admissible region for
$p=0.6$ is the sum of shaded weakly admissible regions for $p=0.2$ and 0.4 plus
the region shaded for $p=0.6$.}
\label{thrd}
\end{figure}
\begin{figure}[H]
 \centering
\includegraphics[scale=0.32]{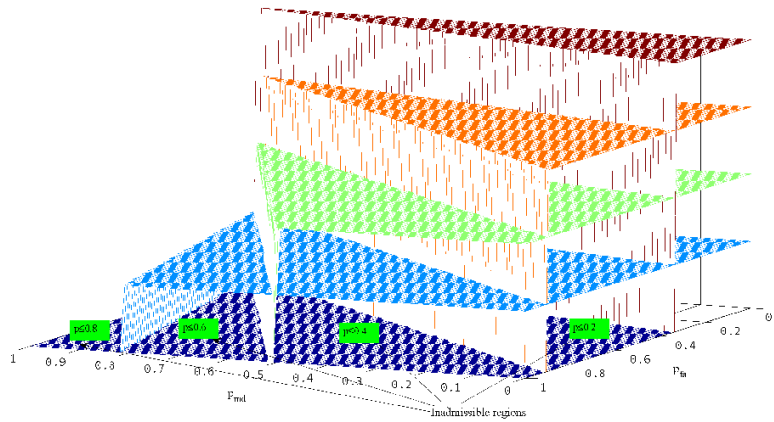}
\caption{$p_{fa}$ vs $p_{md}$ plot showing inadmissible regions.}
\label{pffm}
\end{figure}
\begin{figure}[H]
 \centering
\includegraphics[scale=0.35]{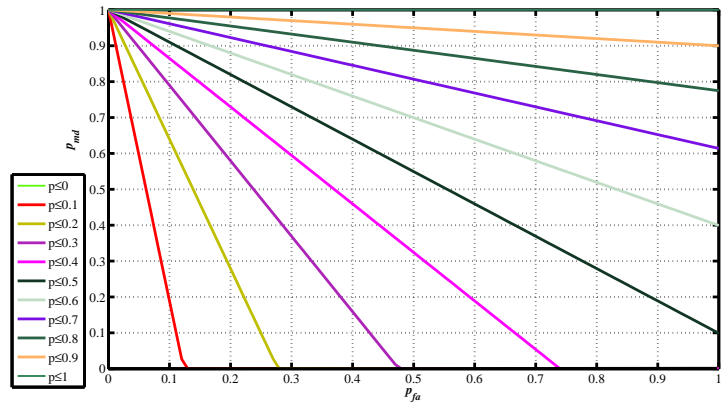}
\caption{$p_{fa}$ vs $p_{md}$ plot for a fixed $RS = 0 dB$ for various values of
$p$.}
\label{pfpm}
\end{figure}
\begin{figure}[H]
 \centering
\includegraphics[scale=0.4]{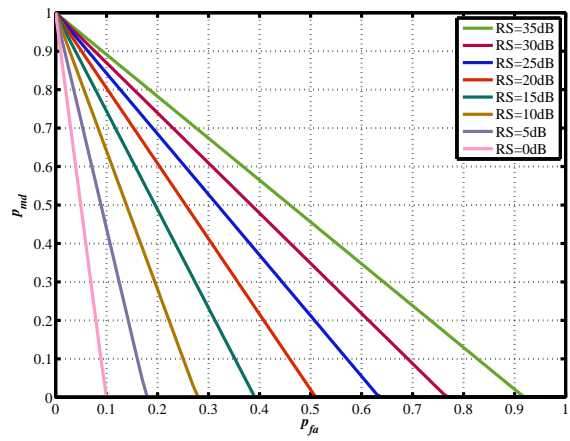}
\caption{$p_{fa}$ vs $p_{md}$ plot for a fixed $p=0.5$ for various values of
relative power levels of PU and CR, $RS$.}
\label{snrvar}
\end{figure}
\begin{figure}[H]
\centering%
\includegraphics[scale=0.37]{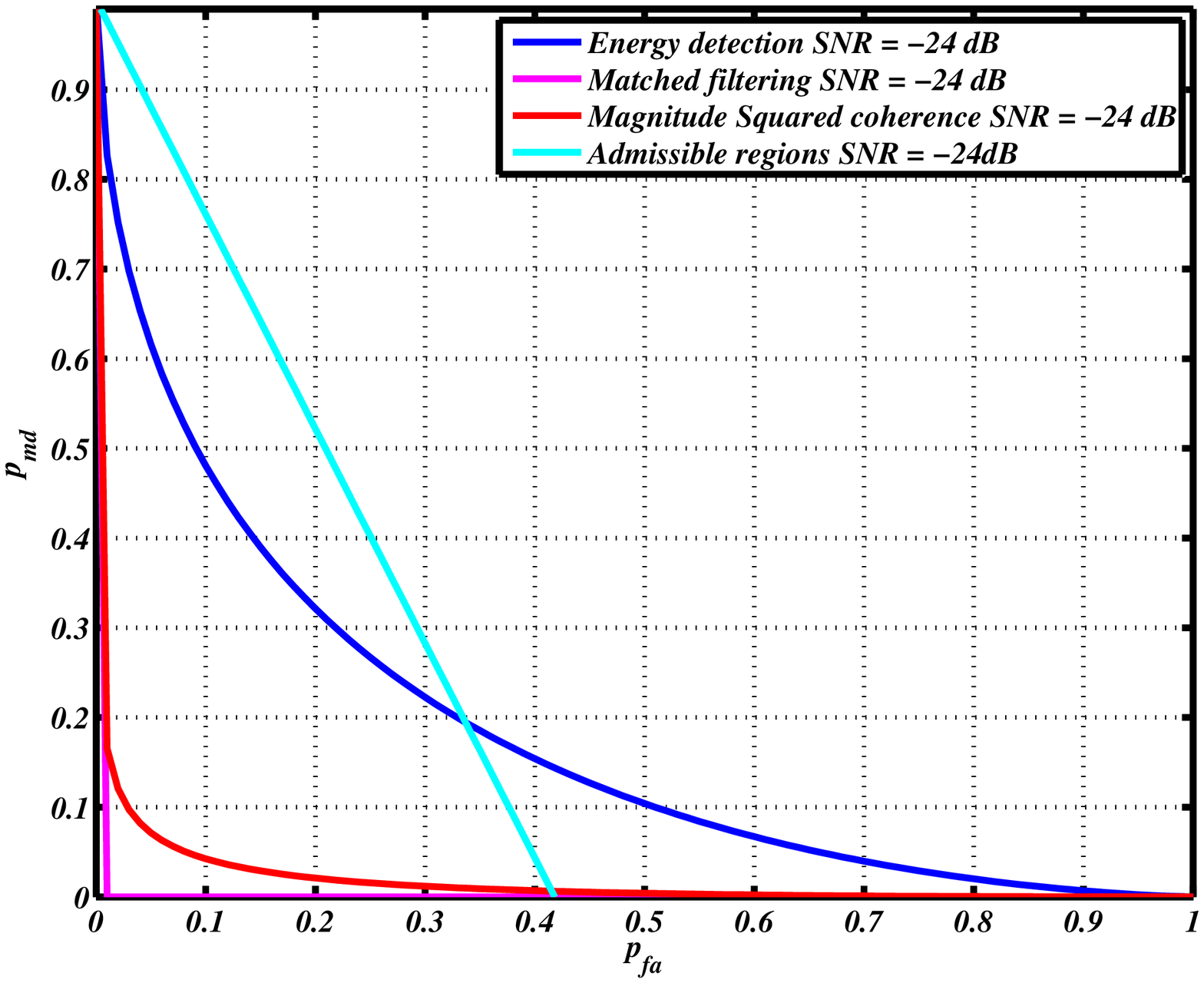}
\caption{Admissible regions for $p=0.2$ and $RS = 20 dB$ for the energy
detector, magnitude squared coherence and matched filter. The average received
SNR is $-24dB$}
\label{snrdet24}
\end{figure}
\begin{figure}[H]
\centering%
\includegraphics[scale=0.37]{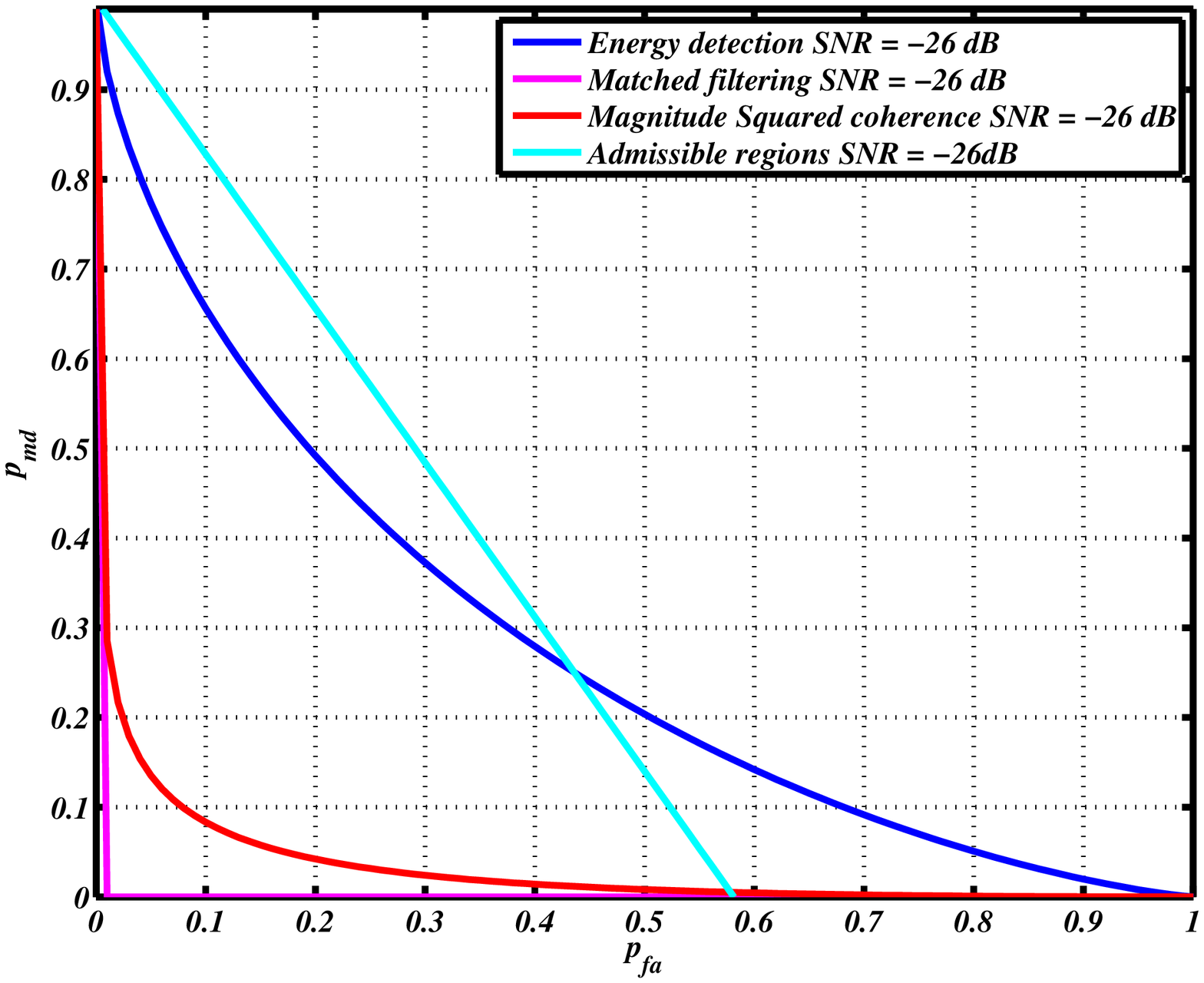}
\caption{Admissible regions for $p=0.2$ and $RS = 20 dB$ for the energy
detector, magnitude squared coherence and matched filter. The average received
SNR is $-26dB$}
\label{snrdet26}
\end{figure}
\bibliographystyle{ieeetr}	
\bibliography{myrefs}

\end{document}